\documentclass[10pt, final, journal, letterpaper, twocolumn]{IEEEtran}

\usepackage[dvips]{graphicx}
\usepackage{times}
\usepackage{cite}
\usepackage{amsmath}
\usepackage{array}
\usepackage{amssymb}

\usepackage{stfloats}
\usepackage{rotating,threeparttable,booktabs}
\usepackage{bm}
\usepackage{dcolumn,booktabs}
\usepackage{multirow}
\usepackage{subfigure}
\usepackage{color}
\usepackage{algorithm}
\usepackage{algorithmic}

\usepackage{amsmath,amsthm,amssymb,amsfonts}
\makeatletter
\thm@headfont{\sc}
\makeatother
\newtheorem{theorem}{Theorem}

\newtheorem{lemma}{Lemma}

\newtheorem{prop}{Proposition}

\usepackage{url}

\begin{document}

\title{QoS-Driven Resource Optimization for Intelligent Fog Radio Access Network: A Dynamic Power Allocation Perspective}


\author{\authorblockA{Jun Yu, Rui Wang, \textit{Senior Member, IEEE}, Jun Wu, \textit{Senior Member, IEEE}} 
\thanks{J. Yu, and R. Wang are with the School of Electronics and Information Engineering at Tongji University, Shanghai, 201804, P. R. China (e-mails: 1610986@tongji.edu.cn, ruiwang@tongji.edu.cn). }
\thanks{J. Wu is with the School of Computer Science and Technology at Fudan University, Shanghai, 201203, P. R. China (e-mail: wujun@fudan.edu.cn).}
}
\maketitle

\begin{abstract}
The fog radio access network (Fog-RAN) has been considered a promising wireless access architecture to help shorten the communication delay and relieve the large data delivery burden over the backhaul links. However, limited by conventional inflexible communication design, Fog-RAN cannot be used in some complex communication scenarios. 
In this study, we focus on investigating a more intelligent Fog-RAN to assist the communication in a high-speed railway environment. Due to the train's continuously moving, the communication should be designed intelligently to adapt to channel variation.  
%
Specifically, we dynamically optimize the power allocation in the remote radio heads (RRHs) to minimize the total network power cost considering multiple quality-of-service (QoS) requirements and channel variation. The impact of caching on the power allocation is considered. The dynamic power optimization is analyzed to obtain a closed-form solution in certain cases. The inherent tradeoff among the total network cost, delay and delivery content size is further discussed. To evaluate the performance of the proposed dynamic power allocation, we present an invariant power allocation counterpart as a performance comparison benchmark. The result of our simulation reveals that dynamic power allocation can significantly outperform the invariant power allocation scheme, especially with a random caching strategy or limited caching resources at the RRHs.
\end{abstract}


\section{Introduction}
As an evolution of the conventional cloud radio access network (C-RAN) \cite{Zhao_JSAC_2016, Yan_access2018}, the fog radio access network (Fog-RAN) has been regarded as a promising new wireless access network architecture, which can help relieve the large data traffic burden in the backhaul links of a cellular network and satisfy the stricter delay requirements in beyond-fifth-generation (B5G) and sixth-generation (6G) wireless communications \cite{Pen_access2016, Zhao_WCM_2020}. The key technology of the Fog-RAN to achieve this performance enhancement is the introduction of caching resources on the edge devices, i.e., the remote radio heads (RRHs), of the access network. With some popular contents cached at the RRHs, the frequently requested data by the terminals can be instantly acquired from the RRHs without needing to be fetched from remote servers. This scenario avoids possible traffic congestion over the backhaul links and further shortens the content delivery delay.

Due to the advantage of the Fog-RAN, the corresponding studies have recently received much attention. Current studies on the Fog-RAN mainly
focus on two aspects: performance characterization and resource optimization. Here, the resources include caching, beamforming, power and computation resources. The authors in \cite{Jiang_TWC_2020} investigated the delay and energy efficiency of the Fog-RAN by considering the hybrid caching strategy, which combines coded cached, nonpartitioned cached and uncached files. The hybrid caching strategy was further optimized to balance the delay and energy efficiency. Successful transmission probability (STP) is often used a performance metric to characterize the impact of content caching at an RRH. In \cite{Wang_TVT2019}, the authors first derived the STP of the Fog-RAN system with a proactive probabilistic caching strategy. The caching probability for different contents was further optimized to maximize the STP. In \cite{Wang_CL_2019}, Fog-RAN-assisted transmission was studied in a heterogeneous Fog-RAN wireless network. Different from \cite{Wang_TVT2019}, the authors in \cite{Wang_CL_2019} optimized the caching strategy in both the RRHs and mobile users, aiming to optimize the STP. The power allocation of the Fog-RAN was studied in \cite{Bai_2019_access, Rahman_TVT_2020, Xiang_TVT_2020, Zhang_AT-RASC2018}. In \cite{Bai_2019_access}, the authors proposed using the nonorthogonal multiple access (NOMA) technique to distinguish multiple users in the Fog-RAN. An improved fractional transmit power allocation algorithm was developed for the three-user NOMA-assisted Fog-RAN system. In \cite{Rahman_TVT_2020}, the latency minimization problem for the Fog-RAN was formulated with a dynamic user demand. To understand the demand of the user and to intelligently perform the joint optimization of the proactive cache strategy and power allocation, a deep reinforcement learning approach was used. The authors in \cite{Xiang_TVT_2020} jointly optimized the power and the model selection for uplink transmission of the Fog-RAN. The reinforcement learning approach was used to solved the nonconvex mixed-integer programming problem. Subchannel assignment and power control were jointly optimized in \cite{Zhang_AT-RASC2018} for a mmWave-based Fog-RAN. The alternative direction method of multipliers (ADMM) was utilized to solve the power control problem.

The beamforming design problem for the Fog-RAN was investigated in \cite{He_TCOM_2019, Tao_TWC_2016, ErkaiTWC_2018}. In particular, the uplink Fog-RAN was studied \cite{He_TCOM_2019}, aiming to maximize the delivery rate subject to the constraints, including the backhaul capacity, transmit power, and file size. A two-layer transmission scheme including the cache level and network level was applied for content transmission. Both the centralized algorithm and the decentralized algorithm were proposed for optimizing the beamforming problem. In \cite{Tao_TWC_2016}, the authors considered the Fog-RAN with multicast transmission. The beamforming design problem was constructed to minimize the network total cost subject to the power and signal-to-interference-plus-noise ratio (SINR) constraints. The sparse beamforming vector was found with the convex-concave procedure. The authors in \cite{ErkaiTWC_2018} further extended the channel model in \cite{Tao_TWC_2016} to a scenario with both multicast and unicast transmission. A branch-and-bound algorithm was utilized to find the global optimal solution.

Computational resource allocation was considered in \cite{Ma_Acess_2020, Khumalo_2020, Li_2020_TWC, Dang_jSAC_2019} for performance improvement. In \cite{Ma_Acess_2020}, edge/cloud computing and edge computing task migration were included in the design problem to improve the quality-of-service (QoS) at the users. Optimizing the user association and computational offloading, the communication resources and computational resources could be balanced. A reinforcement-learning-based approach was used in \cite{Khumalo_2020} to optimize the computational resources with the objective of reducing the latency and energy consumption. In \cite{Li_2020_TWC}, the authors proposed using the computational resources at the edge devices of the Fog-RAN to create cooperative downlink transmission to decrease the latency. An order-optimal upload-download communication latency pair was characterized to evaluate the network performance. In \cite{Dang_jSAC_2019}, the authors formulated a joint decision problem for communication, caching and computing resources. The problem was modeled as a multiple-choice multidimensional knapsack problem, and was then solved by the Lagrangian dual decomposition approach.

As we summarized above, even though a large number of contributions to the Fog-RAN have been reported, the investigations are still not sufficient to address all challenges under different application scenarios. 
%
A key shortcoming of the existing work is that current studied Fog-RAN  are not intelligent enough, leading to that it cannot be used in some complex dynamic communication scenarios. A typical scenario is the 
high-speed railway scenario \cite{Ai_ComM2015}. The technology of the high-speed railway, identified as a typical wireless communication scenario for future cellular networks, has developed rapidly on a global scale \cite{Fan_Access2016}. How to use the Fog-RAN architecture to provide reliable and low-latency communication services is an interesting topic that will require in-depth studies.

To compensate for the deficiencies of the current research on the Fog-RAN, we aim to design a more intelligent Fog-RAN to assist the communication in a high-speed railway environment. 
In specific, we investigate the Fog-RAN-assisted downlink data transmission in the high-speed railway scenario. Because of the time-varying channel state information, the corresponding studies become more challenging. Although wireless communication in the high-speed railway scenario has received much attention \cite{Muneer_TWC2015, Wang_JSAC2012, Li_Access2017, Zhang_TVT2015, Liu_Access2018}, to our knowledge, few works have studied the Fog-RAN in the high-speed railway scenario. This lack motivates the study of this work.

In the considered high-speed railway Fog-RAN, we assume that the train is cooperatively served by multiple RRHs. By taking the time-varying channel into account, we investigate an intelligently dynamic power allocation to minimize the cost of the entire network power, which includes the transmit power at the RRHs and the power consumed over the backhaul links, subject to a few quality-of-service (QoS)
requirements. With this goal, the caching placement on the RRHs affects the physical layer power allocation. To find a reasonable solution to the considered nonconvex problem, smoothed $l_0$-norm approximation is employed to convert the nonconvex problem into a tractable form. By analyzing the unique channel properties of the considered high-speed railway Fog-RAN, we provide a closed-form solution in certain special cases. With the obtained solution, the inherent tradeoff among the total network cost, delay and delivery content size is further discussed. Moreover, the invariant power allocation counterpart is derived to evaluate the performance of the proposed dynamic power allocation. Simulation results reveal that the intended dynamic power allocation is able to significantly outperform the invariant power allocation scheme. The performance gain is more pronounced when the random caching strategy is used or the caching resources at the RRHs are limited.

The organization of the remainder of the paper is shown below. In Section \ref{channel}, we present the channel model as well as problem formulation of the considered high-speed railway Fog-RAN. In Section \ref{Dynamic}, optimization of the dynamic power allocation is discussed. In Section \ref{Tradeoff}, we analyze the tradeoff among the total network cost, delay and delivery content size. In Section \ref{Invariant}, we present the invariant power allocation counterpart. In Section \ref{simulation}, extensive numerical results are stated to demonstrate the performance. Finally, we conclude the paper in Section \ref{conclusion}.

\section{Channel Model and Problem Formulation}\label{channel}

\begin{figure}
\includegraphics[width= 3.5 in]{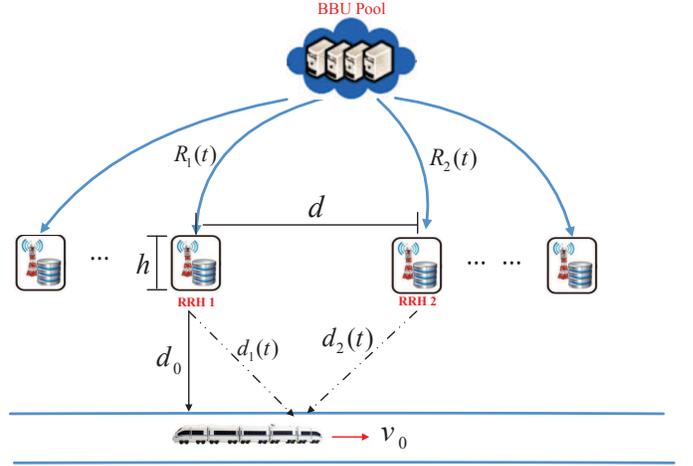}
\caption{Fog-RAN-assisted high-speed communication system.} \label{system_model}
\end{figure}

In this section, we present the channel model of the considered Fog-RAN. After that, the dynamic optimization problem is formulated.

\subsection{Channel model}
As illustrated in Fig. \ref{system_model}, we consider a dynamic wireless transmission scenario of a train running at high speed along a straight-line railway. To provide the required high data rate and low latency downlink transmission, the train is served by a Fog-RAN in which multiple base stations, i.e., RRHs, are uniformly deployed along one side of the road at intervals of the same size. The RRHs are assumed to be connected to the same baseband unit (BBU). The BBU is responsible for the RRH resource allocation as well as the serving-RRH switch. At any time, the train is served by the two nearest RRHs. Because the service repeats periodically for different pairs of neighboring RRHs, with no loss of generality, only one time interval is considered, where the train is assisted by RRH $1$ and RRH $2$.

It is assumed that the contents requested by the train include $L$ different popular contents of the same size $Q$.
All RRHs have a limited local storage size, and all $L$ contents cannot be cached in a single RRH simultaneously. Suppose $F_n$ as the local storage size of RRH $n$; thus, we have $F_n< LQ$. According to the requested frequency of the contents, the popularity distribution of the contents is denoted by ${\bf p}=[p_1, p_2, \cdots, p_L]$, where $p_l\in (0, 1)$ denotes the popularity of content $l$. With unchanged generality, in this study, we make the assumption that $p_1\geq p_2 \geq \cdots \geq p_L$ and that the content popularity follows a Zipf distribution given by $p_l = \frac{l^{-\eta}}{\sum^L_{l=1} l^{-\eta} }$, where $\eta$ denotes the shaping parameter defining the skewness of the popularity distribution \cite{Tao_TWC_2016, Wang_TVT_2019,Wang_CL_2019}. According to the content popularity distribution, the RRHs can store the requested contents with different caching strategies. Two caching strategies that widely used are the popularity-aware caching (PopC) strategy and random caching (RndC) strategy. The PopC strategy allows each RRH to cache the most popular contents until its storage size is fully utilized, while the RndC strategy asks each RRH to cache the contents randomly with identical probabilities regardless of the content popularity distribution. Therefore, to identify the used caching strategy, we define a cache placement matrix $C \in \mathbb{{B}}^{N \times L}$ with $c_{n,l}\in {0, 1}$. Specifically, $c_{n,l} = 1$ occurs if content $l$ is cached in RRH $n$; otherwise, content $l$ is not cached in RRH $n$. To satisfy the RRH storage size constraints, we have $\sum_{l = 1}^{F} c_{n,l} Q \leq F_n, \forall n$.

We consider the Fog-RAN-assisted downlink transmission described in Fig. \ref{system_model}. We assume $d$ equal intervals between every two RRHs that $d_0$ is the distance between each RRH and the road, and that  $h$ is the height of the transmit antenna on each RRH. The train moves down the railway at a unchanging velocity $v$. To determine the coordinates of the RRHs, we assume that there is an original point $o$ and the system time when the train passes through $o$  equals to $0$. The coordinates of RRH $n$ are represented as $(l_n, d_0)$. During the time interval $(0, T]$, RRH $1$ and $2$  serve the train. According to the geometric structure of the system, the distance between RRH $n$ and the train can be denoted by
\begin{equation}\label{eqn_channel_model_0}
d_n(t) = \sqrt{(v t - l_n)^2 + d_0^2+h^2}, ~~t \in (0, T].
\end{equation}
When $t>T$, the BBU coordinates the handoff process, and the train is served by a new set of RRHs. However, because the communications over different time intervals are periodic, with no loss of generality, we next focus only on the dynamic resource allocation over time interval $t\in (0, T]$.

During the time interval $t \in (0, T]$, we assumed that a requested content has to be delivered from the RRHs to the train. Denote $x(t)$ as the modulated signal for the requested content transmitted in the downlink. Here, we assume that signals transmitted from different RRHs are sent over an orthogonal bandwidth. Moreover, assume that signal $x(t)$ is a stochastic process with mean of zero and with unit variance. Thus, the baseband signal transmitted from RRH $n$ at time $t$ can be represented by
\begin{equation}\label{eqn_channel_model_1}
y_n(t) = \sqrt{P_n(t)} h_n(t)x(t) + n_n(t),
\end{equation}
where $P_n(t)$ denotes the instantaneous transmit power at RRH $n$, $h_n(t)$ shows the instantaneous channel coefficient, and $n_n(t)$ represents the additive complex cycle symmetric Gaussian noise at the train following $CN(0, \sigma^2)$. In this study, we assume that the train runs in an open area and that the channel coefficient is dominated by the line-of-sight (LOS) component without any scatter. In this way, the propagation attenuation model can be represented as $h_n(t)=\frac{\sqrt{G}}{d^\alpha_n(t)}$, where $G$ shows the constant channel gain and $\alpha$ represents the path-loss exponent.

At the receiver, the received signals can be combined using a maximal ratio combiner. The instantaneous achievable rate at time $t$ can be given as
\begin{equation}\label{eqn_channel_model_2}
C(t) = B \log_2 \left(1+\sum_{n\in \mathcal{N}} \frac{P_n(t) |h_n(t)|^2}{\sigma^2} \right),
\end{equation}
where $\mathcal{N}=\{1,2\}$ and $B$ denotes the frequency bandwidth that is distributed for every channel between an RRH and the train.

\subsection{Problem formulation}

In the considered Fog-RAN-assisted downlink transmission, if the content requested by the train has been cached at the serving RRH, the serving RRH is able to instantly convey the content to the train; if not, fetching the content from the BBU via backhaul links is required for the RRH , which consumes extra backhaul transmission resources. Denote the instantaneous content delivery rate over the backhaul link between RRH $n$ and the train at time $t$ as $R_n(t)$. The target of dynamic power allocation is to reduce the entire network power cost as much as we can, which includes the backhaul power consumption and RRH convey power consumption. Specifically, assuming that content $l$ is requested by the train, the backhaul power consumption can be represented as
\begin{equation}\label{eqn_problem_1}
{\rm Cost}_b = \int_0^T \sum^2_{n=1} \beta  ||\int_0^T P_n(t)dt||_0 (1-c_{n,l})R_n(t) d t.
\end{equation}
Here, we consider that only the backhaul link associated with the active RRH at which the requested content is not cached consumes extra power. In \eqref{eqn_problem_1}, we use the term $||\int_0^T P_n(t)dt||_0$ to indicate the active RRH, which implies that if $P_n(t)$ is not always zero over time period $(0, T]$, RRH $n$ is an active RRH. $1-c_{n,l}$ is used to indicate the impact of content $l$ caching on the backhaul link associated with RRH $n$. We observe that if $c_{n,l}=1$, which means that content $l$ has been at RRH $n$, then $1-c_{n,l}=0$ indicates that no extra power is needed for backhaul link $n$; otherwise, extra power is required for backhaul link $n$.
Parameter $\beta$ in \eqref{eqn_problem_1} is a constant number establishing a relationship between the rate $R_n(t)$ and the cost of power.

In addition, the entire RRH transmit power consumed over time period $(0, T]$ can be denoted as
\begin{equation}\label{eqn_problem_2}
{\rm Cost}_p = \int_0^T \sum_{n \in \mathcal{N}} P_n(t) dt.
\end{equation}
In this way, the network power cost in total can be written as
\begin{equation}\label{eqn_problem_3}
{\rm Cost} = {\rm Cost}_b+{\rm Cost}_p.
\end{equation}

In addition to minimizing the total network power consumption, our dynamic power allocation also considers several QoS-related constraints. The most important one is the delay requirement. Basically, for an RRH that does not cache the requested content, the delay contains two hops: one for the backhaul link and another for the wireless transmission link between this RRH and the train. However, as we here assume that $R_n(t)\geq C(t)$ always succeeds, the instantaneous delay can be written as
\begin{equation}\label{eqn_problem_4}
\tau (t) = \frac{1}{ C(t)}.
\end{equation}

In brief,the overall dynamic power allocation problem is formulated as
\begin{subequations}\label{eqn_problem_5}
\begin{eqnarray}
&&\min\limits_{P_n(t)} {\rm Cost} \label{eqn_problem_5a} \\
s.t. && \frac{1}{T}\int_0^T P_n(t) dt \leq P_{n,{\rm avg}}~~ \forall n \in \mathcal{N} \label{eqn_problem_5b}\\
&&  \tau (t)  \leq \tau_{\rm max} \label{eqn_problem_5c}\\
&& \int_0^T C(t) dt  \geq Q \label{eqn_problem_5d}\\
&& P_n(t)\geq 0, \label{eqn_problem_5f}
\end{eqnarray}
\end{subequations}
where \eqref{eqn_problem_5b} denotes the average power constraint of every RRH, with $P_{n,{\rm avg}}$ being the maximum average power at RRH $n$; \eqref{eqn_problem_5c} denotes the instantaneous transmission delay requirement, with $\tau_{\rm max}$ being the maximum delay requirement;  \eqref{eqn_problem_5d} is the requested content delivery that requires to be completed through the network during the time period $(0,T]$; and \eqref{eqn_problem_5f} shows that the instantaneous power at each time $t$ should not be negative.
Our final objective is to minimize the total network cost through maximizing the instantaneous power at every RRH and the instantaneous transmission rate over each backhaul link.


\section{Dynamic Power Allocation for High-Speed Railway Fog-RAN}\label{Dynamic}

In this section, our solution to \eqref{eqn_problem_5} is presented. Different from the widely studied conventional static power allocation problem, our considered dynamic is more challenging. The reason includes the following aspects. First, because our optimization problem considered the backhaul consumption, a nonconvex $l_0$-norm function is introduced, which results in the nonconvexity of our problem. Moreover, our optimization variable is a function with respect to time $t$, and the constraints in problem \eqref{eqn_problem_5} include an integration form. Basically, the dynamic optimization problems are widely applied in the fields of smart power systems, robotics, etc. \cite{Simonetto_Proceedings_2020}. In what follows, by using certain approximations, we present several ways to find the optimal solution to the approximated problem.

To address the nonconvex objective function in \eqref{eqn_problem_5}, we utilize the continuous smooth log-function to approximate the $l_0$-norm function as \cite{Tao_TWC_2016}
\begin{equation}\label{eqn_approx}
||x||_0 \approx \frac{\log\left(\frac{x}{\theta}+1\right)}{\log\left(\frac{1}{\theta}+1\right)},
\end{equation}
where the function of $\theta$ is to control the smoothness of the approximation. A smaller value of $\theta$ results in a better approximation, while leads to a worse smooth function, vice versa. With \eqref{eqn_approx}, the nonconvex part ${\rm Cost}_b$ in \eqref{eqn_problem_5} can be written as
\begin{equation}\label{eqn_OptPower_2}
\begin{aligned}
{\rm Cost}_b \approx & c  \beta  \sum^N_{n=1} (1-c_{n,f})\log\left(\frac{\int_0^T P_n(t)dt+\theta}{\theta}\right) \\
& \times  \int_0^T R_n(t) d t ,
\end{aligned}
\end{equation}
where $c=\frac{1}{\log\left(\frac{1}{\theta}+1\right)}$. With \eqref{eqn_OptPower_2}, the objective function in \eqref{eqn_problem_5} can be rewritten as
\begin{equation}\label{eqn_OptPower_3}
\begin{aligned}
{\rm Cost}_{\rm appro 1} \approx &  \int_0^T \sum_{n \in \mathcal{N}} P_n(t) dt + \sum_{n \in \mathcal{N}} b_n \\
& \times \log\left(\frac{\int_0^T P_n(t)dt+\theta}{\theta}\right),
\end{aligned}
\end{equation}
where $b_n= c  \beta   (1-c_{n,f})\int_0^T R_n(t) d t$.


It is observed that the approximated function ${\rm Cost}_{\rm appro 1}$ is still nonconvex with respect to $P_n(t)$ as it involves the summation of concave functions $\log\left(\frac{\int_0^T P_n(t)dt+\theta}{\theta}\right)$. To address this problem, we apply the majorization-minimization (MM) theory to find the upper bound of ${\rm Cost}_{\rm appro 1}$. Considering that the logarithmic function is a concave function, the MM theory applied in our case involves using its first-order Taylor expansion as the upper bound. Then, in the MM algorithm, a solution to \eqref{eqn_problem_5} is generated through minimizing the following upper-bounded function:
\begin{equation}\label{eqn_OptPower_4}
\begin{split}
 & {\rm Cost}_{\rm appro 2} \approx    \int_0^T \sum_{n \in \mathcal{N}} P_n(t) dt +
 \sum_{n \in \mathcal{N}} b_n
 \bigg[  \\
 &  \log\Big( \frac{\theta + \int_0^T P^0_n(t) dt}{ \theta }\Big)+ \frac{ \int_0^T P_n(t) dt-\int_0^T P^0_n(t) dt}{\theta + \int_0^T P^0_n(t) dt}  \bigg]
\end{split},
\end{equation}
where $ \int_0^T P^0_n(t) dt$ denotes a basis point of the Taylor expansion of $\log\left(\frac{\int_0^T P_n(t)dt+\theta}{\theta}\right)$.

Assuming $k_n = 1+ b_n \frac{1}{\theta +  \int_0^T P^0_n(t) dt}$,solving \eqref{eqn_problem_5}  finally reduces to solving the following convex problem:
\begin{subequations}\label{eqn_OptPower_5}
\begin{eqnarray}
&&\min\limits_{P_n(t)} \int_0^T \bigg(\sum_{n \in \mathcal{N}} k_n P_n(t) \bigg) dt  \label{eqn_OptPower_5a} \\
s.t. && \frac{1}{T}\int_0^T P_n(t) dt \leq P_{n,{\rm avg}}~~ \forall n \in \mathcal{N} \label{eqn_OptPower_5b}\\
&&  C(t) \geq \frac{1}{\tau_{\rm max}} \label{eqn_OptPower_5c}\\
&& \int_0^T C(t) dt  \geq Q  \label{eqn_OptPower_5d}\\
&&   P_n(t) \geq 0. \label{eqn_OptPower_5f}
\end{eqnarray}
\end{subequations}

For solving \eqref{eqn_OptPower_5}, firstly, we obtain the following theorem.

\begin{theorem}\label{Lemma_1}
If $\frac{T}{\tau_{\rm max}} \geq Q$, we have $C(t)=\frac{1}{\tau_{\rm max}}$ at the optimal solution, and the optimal solution of \eqref{eqn_OptPower_5} can be represented as
\begin{equation}\label{eqn_OptPower_55}
P_1(t)  = \tilde{a}_0(t) -  \tilde{a}_2(t) P_2(t),
\end{equation}
where $\tilde{a}_n(t) $ is as defined in \eqref{eqn_OptPower_7}. Denote the critical time points $\{t^{\prime}, t^{\prime  \prime}\}$ and $\tilde{t}^\prime$ as defined in \eqref{eqn_OptPower_6new11} and \eqref{eqn_OptPower_6new13}, respectively; the optimal $P_2(t)$ is given in the following four cases:
\begin{itemize}
  \item If $\int_0^T \tilde{a}_3(t) dt \leq T P_{2,{\rm avg}}$ and $B\leq 0$
\begin{equation}\label{eqn_OptPower_6new5_lemma}\nonumber
P^*_2(t) = \left\{
               \begin{array}{cc}
                 0 & 0<t < t^{\prime}\\
                 \tilde{a}_3(t) & t^{\prime} \leq t\leq T \\
               \end{array} \right.
\end{equation}
\item If $\int_0^T \tilde{a}_3(t) dt \leq T P_{2,{\rm avg}}$ and $B> 0$
\begin{equation}\label{eqn_OptPower_6new11_lemma}\nonumber
P^*_2(t) = \left\{
               \begin{array}{cc}
                 0 & 0<t< \min \{t^{\prime}, t^{\prime  \prime}\} \\
                 \tilde{a}_3(t) & \min \{t^{\prime}, t^{\prime  \prime}\}\leq t \leq T \\
               \end{array} \right.
\end{equation}
\item If $\int_0^T \tilde{a}_3(t) dt > T P_{2,{\rm avg}}$ and $B\leq 0$
\begin{equation}\label{eqn_OptPower_6new14_lemma}
P^*_2(t) = \left\{
               \begin{array}{cc}
                 0 & 0<t< \max \{t^{\prime}, \tilde{t}^\prime\} \\
                 \tilde{a}_3(t) &  \max \{t^{\prime}, \tilde{t}^\prime\}\leq t \leq T \\
               \end{array} \right.
\end{equation}
\item If $\int_0^T \tilde{a}_3(t) dt > T P_{2,{\rm avg}}$ and $B\leq 0$
\begin{itemize}
  \item If $t^{\prime} \geq \tilde{t}^\prime$, the optimal solution is given by \eqref{eqn_OptPower_6new14_lemma}.
  \item Otherwise, the optimal solution can be obtained by solving linear program problem \eqref{eqn_OptPower_6new15}.
\end{itemize}
\end{itemize}

\end{theorem}
\begin{proof}
It is noted that if $\frac{T}{\tau_{\rm max}} \geq Q$ succeeds, for condition \eqref{eqn_OptPower_5d}, we have $\int_0^T C(t) dt \geq \int_0^T \frac{1}{\tau_{\rm max}}dt = \frac{T}{\tau_{\rm max}} \geq Q$. This implies that condition \eqref{eqn_OptPower_5d} in \eqref{eqn_OptPower_5} is redundant. Next, we show that at the optimal solution, \eqref{eqn_OptPower_5c} must be active. We prove this result using a contradiction. Assume that at the optimal solution of \eqref{eqn_OptPower_5}, the optimal $P_n(t)$ makes $C(t)>\frac{1}{\tau_{\rm max}}$ at some time $t$. Now, we can multiply $P_n(t)$ by a coefficient $\delta_n(t)\in (0, 1)$, which can further reduce the value of the objective function while not violating the constraints. According to the above analysis, under the condition of $\frac{T}{\tau_{\rm max}} \geq Q$, problem \eqref{eqn_OptPower_5} is equivalent to
\begin{subequations}\label{eqn_OptPower_6}
\begin{eqnarray}
&&\min\limits_{P_n(t)} \int_0^T \bigg(\sum_{n \in \mathcal{N}} k_n P_n(t) \bigg) dt   \label{eqn_OptPower_6a} \\
s.t. && \int_0^T P_n(t) dt \leq  T P_{n,{\rm avg}}~~ \forall n \in \mathcal{N} \label{eqn_OptPower_6b} \\
&& B \log_2\Big(1+ \sum_{n \in \mathcal{N}} \frac{G P_n(t)} {d_n(t)^\alpha \sigma^2} \Big) = \frac{1}{\tau_{\rm max}} \label{eqn_OptPower_6c}\\
&&    P_n(t) \geq 0, \forall n. \label{eqn_OptPower_6d}
\end{eqnarray}
\end{subequations}
With the equality constraint \eqref{eqn_OptPower_6c}, by denoting $a_n(t) = \frac{G} {d_n(t)^\alpha \sigma^2}$, we have
\begin{equation}\label{eqn_OptPower_7}
P_1(t) = \tilde{a}_0(t) -  \tilde{a}_2 (t) P_2(t),
\end{equation}
where $\tilde{a}_0(t) =\frac{2^{\frac{1}{B \tau_{\rm max} }}-1}{a_1(t)} $ and $\tilde{a}_2 (t) = \frac{a_2(t)}{a_1(t)} $. Problem \eqref{eqn_OptPower_6} can be further transformed into
\begin{subequations}\label{eqn_OptPower_6new1}
\begin{eqnarray}
&&\min\limits_{P_2(t)} \int_0^T \bigg(k_2 - k_1 \tilde{a}_2 (t)  \bigg) P_2(t) dt   \label{eqn_OptPower_6new1a} \\
s.t. && \int_0^T P_2(t) dt \leq  T P_{2,{\rm avg}} \label{eqn_OptPower_6new1b} \\
&& B \leq \int_0^T \tilde{a}_2 (t) P_2(t) dt \leq A  \label{eqn_OptPower_6new1c}\\
&&    0 \leq P_2(t) \leq \tilde{a}_3(t), \label{eqn_OptPower_6new1d}
\end{eqnarray}
\end{subequations}
where $A = \int_0^T \tilde{a}_0 (t)dt$, $B = \int_0^T \tilde{a}_0 (t)dt - T P_{1,{\rm avg}}$ and $\tilde{a}_3(t) = \frac{\tilde{a}_0 (t)}{\tilde{a}_2 (t)}$. Constraint \eqref{eqn_OptPower_6new1c} comes from the fact that $0\leq \int_0^T P_1(t) dt \leq  T P_{1,{\rm avg}}$, while \eqref{eqn_OptPower_6new1d} comes from the fact that $P_1(t) \geq 0$.

According to the definitions of $\tilde{a}_3(t)$ and $\tilde{a}_2 (t)$, we observe that if constraint \eqref{eqn_OptPower_6new1d} is satisfied, we have
\begin{equation}\label{eqn_OptPower_6new2}
\int_0^T \tilde{a}_2 (t) P_2(t) dt \leq \int_0^T \tilde{a}_2 (t) \tilde{a}_3(t) dt = A.
\end{equation}
Then, problem \eqref{eqn_OptPower_6new1} can be equivalent to
\begin{subequations}\label{eqn_OptPower_6new3}
\begin{eqnarray}
&&\min\limits_{P_2(t)} \int_0^T \bigg(k_2 - k_1 \tilde{a}_2 (t)  \bigg) P_2(t) dt   \label{eqn_OptPower_6new3a} \\
s.t. && \int_0^T P_2(t) dt \leq  T P_{2,{\rm avg}} \label{eqn_OptPower_6new3b} \\
&& \int_0^T \tilde{a}_2 (t) P_2(t) dt  \geq B  \label{eqn_OptPower_6new3c}\\
&&    0 \leq P_2(t) \leq \tilde{a}_3(t). \label{eqn_OptPower_6new3d}
\end{eqnarray}
\end{subequations}

To find the analytical solution to \eqref{eqn_OptPower_6new3}, we next discuss certain particular cases that help simplify the problem.

\emph{\textbf{Case 1)}} when $\int_0^T \tilde{a}_3(t) dt \leq T P_{2,{\rm avg}}$ and $B\leq 0$ are satisfied.

In this case, constraints \eqref{eqn_OptPower_6new3b} and \eqref{eqn_OptPower_6new3c} are redundant. Problem \eqref{eqn_OptPower_6new3} reduces to
\begin{subequations}\label{eqn_OptPower_6new4}
\begin{eqnarray}
\min\limits_{P_n(t)} && \int_0^T \bigg(k_2 - k_1 \tilde{a}_2 (t)  \bigg) P_2(t) dt   \label{eqn_OptPower_6new4a} \\
s.t. &&    0 \leq P_2(t) \leq \tilde{a}_3(t). \label{eqn_OptPower_6new4d}
\end{eqnarray}
\end{subequations}
It is noted that the optimal solution to \eqref{eqn_OptPower_6new4} depends on the sign of $k_2 - k_1 \tilde{a}_2 (t)$. To minimize the objective function, the optimal solution to \eqref{eqn_OptPower_6new4} is given as
\begin{equation}\label{eqn_OptPower_6new5}
P^*_2(t) = \left\{
               \begin{array}{cc}
                 0 & t\in \{t|k_2 - k_1 \tilde{a}_2 (t)>0\} \\
                 \tilde{a}_3(t) & t\in \{t|k_2 - k_1 \tilde{a}_2 (t)\leq 0\} \\
               \end{array} \right..
\end{equation}

\emph{\textbf{Case 2)}} when $\int_0^T \tilde{a}_3(t) dt \leq T P_{2,{\rm avg}}$ and $B> 0$ are satisfied.

In this case, constraint \eqref{eqn_OptPower_6new3b} is redundant. Problem \eqref{eqn_OptPower_6new3} reduces to
\begin{subequations}\label{eqn_OptPower_6new6}
\begin{eqnarray}
&&\min\limits_{P_2(t)} \int_0^T \bigg(k_2 - k_1 \tilde{a}_2 (t)  \bigg) P_2(t) dt   \label{eqn_OptPower_6new6a} \\
s.t. && \int_0^T \tilde{a}_2 (t) P_2(t) dt  \geq B  \label{eqn_OptPower_6new6c}\\
&&    0 \leq P_2(t) \leq \tilde{a}_3(t). \label{eqn_OptPower_6new6d}
\end{eqnarray}
\end{subequations}

Within the time period $(0, T]$, as the train departs from RRH $1$ and move closer to RRH $2$, we see that functions $k_2 - k_1 \tilde{a}_2 (t)$, $\tilde{a}_2 (t)$, and $\tilde{a}_3(t)$ are decreasing, increasing and decreasing functions with respect to $t$, respectively.
According to this observation, the two critical time points are defined as
\begin{equation}\label{eqn_OptPower_6new7}
\begin{split}
t^\prime & \triangleq \{t| k_2 - k_1 \tilde{a}_2 (t) =0 \} \\
t^{\prime  \prime} & \triangleq \left\{t| \int_0^T \tilde{a}_2 (t) \tilde{a}_3(t) dt =\int_0^T \tilde{a}_0 (t) dt=  B \right \}.
\end{split}
\end{equation}
It is noted that for the time region $(t^\prime, T]$, $P_2(t)$ should be equal to $\tilde{a}_3(t) $ to minimize the value of the objective. Therefore, in the scenario with $t^\prime \geq t^{\prime  \prime}$, the optimal solution is presented as
\begin{equation}\label{eqn_OptPower_6new8}
P^*_2(t) = \left\{
               \begin{array}{cc}
                 0 & t< t^{\prime  \prime} \\
                 \tilde{a}_3(t) & t \geq t^{\prime  \prime} \\
               \end{array} \right. .
\end{equation}
For the scenario with $t^\prime < t^{\prime  \prime}$, it is easy to see that for the time region over $t\in ( t^{\prime  \prime}, T]$, the optimal $P_2(t)$ should be equal to $\tilde{a}_3(t) $ as given in \eqref{eqn_OptPower_6new8}. However, this kind of solution cannot satisfy the constraint \eqref{eqn_OptPower_6new6c}. To this end, we need to find a certain $P_2(t)$ in the time period $t\in (0, t^{\prime  \prime}]$ to satisfy
\begin{equation}\label{eqn_OptPower_6new9}
\int_0^{t^{\prime  \prime}} \tilde{a}_2 (t) P_2(t) dt  =  B - \int_{t^{\prime  \prime}}^T \tilde{a}_0 (t) dt.
\end{equation}
Next we show that the optimal solution over $t\in (0, t^{\prime  \prime}]$ is
\begin{equation}\label{eqn_OptPower_6new10}
P^*_2(t) = \left\{
               \begin{array}{cc}
                 0 & t< t^{\prime} \\
                 \tilde{a}_3(t) & t \in (t^{\prime}, t^{\prime  \prime}] \\
               \end{array} \right. .
\end{equation}
We prove this result with contradiction analysis. Assume that there is a new optimal solution $P^{**}_2(t) $ over $t\in (0, t^{\prime  \prime}]$ given by
\begin{equation}\label{eqn_OptPower_6new11}
P^{**}_2(t) = \left\{
               \begin{array}{cc}
                 0 & t< \tilde{t} \\
                 <\tilde{a}_3(t) & t \in (\tilde{t}, t^{\prime  \prime}] \\
               \end{array} \right. .
\end{equation}
To satisfy \eqref{eqn_OptPower_6new9}, we have $\tilde{t}<t^{\prime}$. However, $k_2 - k_1 \tilde{a}_2 (t)$ is a decreasing function over $t$. Hence, $\int_0^{t^{\prime  \prime}} \bigg(k_2 - k_1 \tilde{a}_2 (t)  \bigg) P^{**}_2(t)dt$ must be larger than $\int_0^{t^{\prime  \prime}} \bigg(k_2 - k_1 \tilde{a}_2 (t)  \bigg) P^{*}_2(t)dt$ with $P^*_2(t)$ given in \eqref{eqn_OptPower_6new10}, which implies that $P^{**}_2(t) $ cannot be an optimal solution. According to the above analysis, the optimal solution for \emph{Case 2} is denoted as
\begin{equation}\label{eqn_OptPower_6new11}
P^*_2(t) = \left\{
               \begin{array}{cc}
                 0 & t< \min \{t^{\prime}, t^{\prime  \prime}\} \\
                 \tilde{a}_3(t) & t\geq \min \{t^{\prime}, t^{\prime  \prime}\} \\
               \end{array} \right. .
\end{equation}

\emph{\textbf{Case 3)}} when $\int_0^T \tilde{a}_3(t) dt > T P_{2,{\rm avg}}$ and $B\leq 0$ are satisfied.

In this case, the constraint \eqref{eqn_OptPower_6new3c} is redundant. Problem \eqref{eqn_OptPower_6new3} reduces to
\begin{subequations}\label{eqn_OptPower_6new12}
\begin{eqnarray}
&&\min\limits_{P_2(t)} \int_0^T \bigg(k_2 - k_1 \tilde{a}_2 (t)  \bigg) P_2(t) dt   \label{eqn_OptPower_6new12a} \\
s.t. && \int_0^T P_2(t) dt \leq  T P_{2,{\rm avg}} \label{eqn_OptPower_6new12b} \\
&&    0 \leq P_2(t) \leq \tilde{a}_3(t). \label{eqn_OptPower_6new12d}
\end{eqnarray}
\end{subequations}
We define $\tilde{t}^\prime$ as
\begin{equation}\label{eqn_OptPower_6new13}
\tilde{t}^\prime \triangleq \left\{t| \int_{\tilde{t}^\prime}^T \tilde{a}_3(t)  dt =  T P_{2,{\rm avg}} \right \}.
\end{equation}
Again, because $k_2 - k_1 \tilde{a}_2 (t)$ is a decreasing function, similar to the analysis in \emph{Case 2}, the optimal solution to problem \eqref{eqn_OptPower_6new12} is presented as
\begin{equation}\label{eqn_OptPower_6new14}
P^*_2(t) = \left\{
               \begin{array}{cc}
                 0 & t< \max \{t^{\prime}, \tilde{t}^\prime\} \\
                 \tilde{a}_3(t) & t\geq \max \{t^{\prime}, \tilde{t}^\prime\} \\
               \end{array} \right. .
\end{equation}
It is noted that in \eqref{eqn_OptPower_6new14}, the critical time point is $\max \{t^{\prime}, \tilde{t}^\prime\}$, which is different from \emph{Case 2} due to the inequality constraint \eqref{eqn_OptPower_6new12b}.

\emph{\textbf{Case 4)}} when $\int_0^T \tilde{a}_3(t) dt > T P_{2,{\rm avg}}$ and $B> 0$ are satisfied.

In this case, according to Lemma \ref{corollary_appendix}, the feasibility of problem \eqref{eqn_OptPower_6new3} requires $\tilde{t}^\prime \leq t^{\prime  \prime}$. Under this condition, if $t^{\prime} \geq \tilde{t}^\prime$, we see that the optimal solution is equal to \eqref{eqn_OptPower_6new11} and that constraint \eqref{eqn_OptPower_6new3b} is redundant. Otherwise, the optimal solution can be approximately obtained through solving the following linear programming:
\begin{subequations}\label{eqn_OptPower_6new15}
\begin{eqnarray}
&&\min\limits_{P_2(t_m)} \sum^M_{m=1} \bigg(k_2 - k_1 \tilde{a}_2 (t_m)  \bigg) P_2(t_m) {\bigtriangleup t}   \label{eqn_OptPower_6new3a} \\
s.t. && \sum^M_{m=1}  P_2(t_m) {\bigtriangleup t} \leq  T P_{2,{\rm avg}} \label{eqn_OptPower_6new3b} \\
&& \sum^M_{m=1}  \tilde{a}_2 (t_m) P_2(t_m) {\bigtriangleup t}  \geq B  \label{eqn_OptPower_6new3c}\\
&&    0 \leq P_2(t_m) \leq \tilde{a}_3(t_m), \forall m. \label{eqn_OptPower_6new3d}
\end{eqnarray}
\end{subequations}
The linear problem is obtained by sampling the time period $t\in(0, T]$ as the discrete time points $\{t_1, t_2,\cdots,t_M\}$ with the adjacent sampling point interval given by ${\bigtriangleup t}$. It is noted that ${\bigtriangleup t}$ is sufficiently small; thus, we can obtain an approximately optimal solution via \eqref{eqn_OptPower_6new15}.

Now, by combining the solutions of \emph{Case 1}-\emph{Case 4}, we obtain the final optimal solution.
This completes the proof of Theorem \ref{Lemma_1}.

\end{proof}

Consider a special case where the train is assisted by only RRH $1$ over the time period $(0, T]$, i.e., $\mathcal{N}=\{1\}$. Theorem \ref{Lemma_1} can be reduced to the following lemma.
\begin{lemma}\label{corollary_1}
If $\frac{T}{\tau_{\rm max}} \geq Q$ and $\mathcal{N}=\{1\}$, we have $C(t)=\frac{1}{\tau_{\rm max}}$ at the optimal solution, and the optimal solution to \eqref{eqn_OptPower_5} can be denoted as
\begin{equation}\label{eqn_OptPower_15}
P_1(t) = \left( 2^{\frac{1}{B\tau_{\rm max}}} -1 \right)\frac{d^{\alpha}_1(t) \sigma^2}{G}.
\end{equation}
\end{lemma}
\begin{proof}
When $\mathcal{N}=\{1\}$, problem \eqref{eqn_OptPower_5} can be written as
\begin{subequations}\label{eqn_OptPower_15}
\begin{eqnarray}
&&\min\limits_{P_1(t)} \int_0^T \bigg( k_1 P_1(t) \bigg) dt  \label{eqn_OptPower_15a} \\
s.t. && \frac{1}{T}\int_0^T P_1(t) dt \leq P_{1,{\rm avg}}~~ \label{eqn_OptPower_15b}\\
&&  C(t) = \frac{1}{\tau_{\rm max}} \label{eqn_OptPower_15c}\\
&&   P_1(t) \geq 0. \label{eqn_OptPower_15f}
\end{eqnarray}
\end{subequations}
If the problem is feasible, the solution is determined by constraint \eqref{eqn_OptPower_15c}, which completes the proof of Lemma \ref{corollary_1}.
\end{proof}

If condition $\frac{T}{\tau_{\rm max}} \geq Q$ is not satisfied, the conclusions shown in Theorem \ref{Lemma_1} and Lemma \ref{corollary_1} will not be applicable. To propose an efficient way to address this problem, we first give the following lemma.
\begin{lemma}\label{corollary_2}
If $\frac{T}{\tau_{\rm max}} < Q$, the optimal solution to \eqref{eqn_OptPower_5} can be represented as
\begin{equation}\label{eqn_OptPower_neweq1}
P_n(t) = P_{n,1}(t) + P_{n,2}(t),
\end{equation}
where $P_{n,1}(t)$ is the optimal solution obtained by solving \eqref{eqn_OptPower_6} and $P_{n,2}(t)$ is the solution to
\begin{subequations}\label{eqn_OptPower_neqeq2}
\begin{eqnarray}
&&\min\limits_{P_{n,2}(t)} \int_0^T \bigg(\sum_{n \in \mathcal{N}} k_n P_{n,2}(t) \bigg) dt  \label{eqn_OptPower_neqeq2a} \\
s.t. && \frac{1}{T}\int_0^T P_{n,2}(t) dt \leq b_n ~~ \forall n \in \mathcal{N} \label{eqn_OptPower_neqeq2b}\\
&&  \sum_{n\in \mathcal{N}}  \kappa_n(t) P_{n,2}(t)  \geq 0 \label{eqn_OptPower_neqeq2c}\\
&& \int_0^T B \log_2 \left(c_n(t) + \sum_{n \in \mathcal{N}} \kappa_n(t) P_{n,2}(t) \right) dt  \geq Q  ~~~~~~~\label{eqn_OptPower_neqeq2d}\\
&&   P_{n,2}(t) \geq - P^*_{n,1}(t), \label{eqn_OptPower_neqeq2f}
\end{eqnarray}
\end{subequations}
where $b_n =P_{n,{\rm avg}}- \frac{1}{T}\int_0^T P^*_{n,1}(t) dt $, $\frac{G}{d^\alpha_n(t)\sigma^2} = \kappa_n(t)$ and $c_n(t) = 1+\sum^N_{n=1} \kappa_n(t)P^*_{n,1}(t)$, with $P^*_{n,1}(t)$ being the optimal solution to $P_{n,1}(t)$.
\end{lemma}
\begin{proof}
It is noted that for any feasible $P_n(t)$ of \eqref{eqn_OptPower_5}, we can always decompose it into a sum of $P_{n,1}(t)$ and $P_{n,2}(t)$. Basically, the choice of $P_{n,1}(t)$ and $P_{n,2}(t)$ can be arbitrary as long as their sum is equal to $P_n(t)$. Without reduction of generality, we assume that the term $P_{n,1}(t)$ is chosen as the optimal solution to \eqref{eqn_OptPower_6}, denoted by  $P^*_{n,1}(t)$. Then, problem \eqref{eqn_OptPower_5} changes to
\begin{subequations}\label{eqn_OptPower_neweq3}
\begin{eqnarray}
&&\min\limits_{P_{n,2}(t)} \int_0^T \bigg(\sum_{n \in \mathcal{N}} k_n (P^*_{n,1}(t)+P_{n,2}(t)) \bigg) dt   \label{eqn_OptPower_neweq3a} \\
s.t. && \frac{1}{T}\int_0^T \bigg( P^*_{n,1}(t)+P_{n,2}(t)\bigg) dt \leq P_{n,{\rm avg}}~~ \forall n \in \mathcal{N} \label{eqn_OptPower_neweq3b}~~~~~~~\\
&&  B \log_2 \left(1+ \sum_{n\in \mathcal{N}}\kappa_n(t)P^*_{n,1}(t)   + \right. \nonumber\\
&& ~~~~~~~~~~~\left. \sum_{n \in \mathcal{N}}\kappa_n(t)P_{n,2}(t) \right) dt \geq \frac{1}{\tau_{\rm max}} \label{eqn_OptPower_neweq3c}\\
&& \int_0^T B \log_2 \left(1+ \sum_{n \in \mathcal{N}}\kappa_n(t)P^*_{n,1}(t) + \right. \nonumber\\
&& ~~~~~~~~~~~~~~~\left. \sum_{n \in \mathcal{N}}\kappa_n(t)P_{n,2}(t) \right) dt  \geq Q  \label{eqn_OptPower_neweq3d}\\
&&   P^*_{n,1}(t)+ P_{n,2}(t)  \geq 0. \label{eqn_OptPower_neweq3f}
\end{eqnarray}
\end{subequations}
Problem \eqref{eqn_OptPower_neweq3} involves only solving variable $P_{n,2}(t)$. We next rewrite the constraints in \eqref{eqn_OptPower_neweq3} by analyzing the value range of $P_{n,2}(t)$. Because we set $B \log_2 \left(1+ \sum_{n \in \mathcal{N}}  \frac{G P^*_{n,1}(t)}{d^\alpha_n(t)\sigma^2}\right)=\frac{1}{\tau_{\rm max}}$, it is observed that $\sum_{n \in \mathcal{N}}  \kappa_n(t) P_{n,2}(t) \geq 0$ must be satisfied; otherwise, the delay constraint \eqref{eqn_OptPower_neweq3c} is violated. Then, by rewriting constraints \eqref{eqn_OptPower_neweq3c} and \eqref{eqn_OptPower_neweq3f} as \eqref{eqn_OptPower_neqeq2c} and \eqref{eqn_OptPower_neqeq2f}, respectively, we have Problem \eqref{eqn_OptPower_neqeq2}. This completes the proof of Lemma \ref{corollary_2}.

\end{proof}

With Lemma \ref{corollary_2}, solving $P_n(t)$ from \eqref{eqn_OptPower_5} under the condition $\frac{T}{\tau_{\rm max}} < Q$ reduces to finding $P_{n,2}(t)$ from \eqref{eqn_OptPower_neqeq2}. The result in Lemma \ref{corollary_2} will be useful in the performance tradeoff analysis in Section \ref{Tradeoff}.

We can easily observe that problem \eqref{eqn_OptPower_neqeq2} is a convex problem. Then, to obtain the optimal solution, we construct an algorithm based on the Karush-Kuhn-Tucker (KKT) conditions. To proceed, firstly, the Lagrangian function is presented as
\begin{equation}\label{eqn_OptPower_18}
\begin{split}
L =  & \int_0^T \bigg(\sum_{n \in \mathcal{N}} \Big( k_n P_{n,2}(t) +  \mu_{1,n} (P_{n,2}(t) - T b_n) \Big) \\
 & -   \mu_2 \Big(B \log_2 (c_n(t) + \sum_{n \in \mathcal{N}} \kappa_n(t)P_{n,2}(t))  - Q \Big) \bigg) dt \\
&  -\mu_3(t) \sum^N_{n=1}  \kappa_n(t) P_{n,2}(t),
\end{split}
\end{equation}
where $\mu_{1,n}$, $\mu_2$ and $\mu_3(t)$ are non-negative multipliers related to the constraints \eqref{eqn_OptPower_neqeq2b}, \eqref{eqn_OptPower_neqeq2d} and \eqref{eqn_OptPower_neqeq2c}, respectively.
To minimize the Lagrangian function, it is necessary to differentiate the Lagrangian function with respect to $P_{n,2}(t)$ and set the derivative to zero for each time $t$, which is,
\begin{equation}\label{eqn_OptPower_19}
\begin{aligned}
 \frac{\partial L} {\partial P_{n,2}(t)} = & k_n +  \mu_{1,n}-\mu_3(t)\kappa_n(t) - \\
&\frac{\mu_2 B}{\log2} \frac{\kappa_n(t) } {c_n(t) + \sum_{n \in \mathcal{N}} \kappa_n(t) P_{n,2}(t)}
= 0.
\end{aligned}
\end{equation}
Through combining with constraint \eqref{eqn_OptPower_neqeq2f}, the solution can be obtained given by
\begin{equation}\label{eqn_OptPower_20}
\begin{aligned}
P_{n,2}(t) = & \left[  \left(\frac{\mu_2 G  B}{\log2 d_n(t)^\alpha \sigma^2 (k_n + \mu_1-\mu_3(t)\kappa_n(t))}  - c_n(t) - \right. \right. \\
& \left. \left. \sum_{m\neq n} \frac{G P_{m,2}(t)} {d_m(t)^\alpha \sigma^2} \right)  \times \frac{d_n(t)^\alpha \sigma^2}{G}, - P^*_{n,1}(t)\right]^+.
\end{aligned}
\end{equation}
\eqref{eqn_OptPower_20} shows that $P_{n,2}(t)$ values with different $n$ are coupled with each other, and the final solution of $P_{n,2}(t)$ can be obtained by iteratively updating them until convergence. During the iteration, Lagrangian multipliers $\mu_{1,n}$, $\mu_2$ and $\mu_3(t)$ can be obtained via the subgradient technique.

Another way to solve \eqref{eqn_OptPower_neqeq2} is to sample the time period to generate certain discrete time points. If the time interval between two adjacent discrete points is small enough, the obtained result can approximately be considered as a solution to \eqref{eqn_OptPower_neqeq2}. Denote the discrete time points as $\{t_1, t_2,\cdots,t_M\}$ and the adjacent sampling point interval as ${\bigtriangleup t}$; the power $P_{n,2}(t_m)$ can be efficiently obtained through solving the following problem:
\begin{subequations}\label{eqn_OptPower_neqeq22}
\begin{eqnarray}
&&\min\limits_{P_{n,2}(t_m)} \sum^M_{m=1}\sum_{n \in \mathcal{N}} k_n P_{n,2}(t_m){\bigtriangleup t}    \label{eqn_OptPower_neqeq22a} \\
s.t. && \sum^M_{m=1} P_{n,2}(t_m){\bigtriangleup t} \leq T b_n ~~ \forall n \in \mathcal{N} \label{eqn_OptPower_neqeq22b}\\
&&  \sum^N_{n=1}  \kappa_n(t_m) P_{n,2}(t_m)  \geq 0 ~~ \forall m\label{eqn_OptPower_neqeq22c}\\
&&  \sum^M_{m=1} B \log_2 \bigg(c_n(t_m) + \sum^N_{n=1} \kappa_n(t_m)   \label{eqn_OptPower_neqeq22d} \\
&&  \times P_{n,2}(t_m) \bigg){\bigtriangleup t}  \geq Q  \nonumber \\
&&   P_{n,2}(t_m) \geq - P^*_{n,1}(t_m) ~ \forall m. \label{eqn_OptPower_neqeq22f}
\end{eqnarray}
\end{subequations}
It is noted that problem \eqref{eqn_OptPower_neqeq22} has only one nonlinear convex constraint \eqref{eqn_OptPower_neqeq22d}, and thus can be efficiently solved by an interior point algorithm using the software CVX \cite{cvx}.

To summarize, the algorithm we proposed to solve \eqref{eqn_OptPower_5} is as follows:

\begin{algorithm}
\caption{\small Dynamic power optimization in problem \eqref{eqn_OptPower_5}}
\begin{itemize}
\item \textbf{Input:} BS interval $d$, speed $v_0$, BS height $h_0$, time interval $(0,T]$, BS coordinate $(l_n, d_0)$, content size $Q$, local storage size $F_n$, cache placement matrix $C$, bandwidth $B$, noise power $\sigma^2$, backhaul cost ratio $\beta$, backhaul rate $R_n(t)$, delay requirement $\tau_{\rm max}$, average power of BS $P_{n,{\rm avg}}$.
\item \textbf{Initialization:} Basis point of the Taylor expansion of $P_n(t)$, that is, $P^0_n(t)$.
\item \textbf{Output:} Dynamic power allocation $P_n(t)$.
\item \textbf{While} not convergence \textbf{do}
\begin{itemize}
  \item Update power $P_n(t)$;
  \begin{enumerate}
    \item Update $P_n(t)$ using Theorem \ref{Lemma_1} if $\frac{T}{\tau_{\rm max}} \geq Q$;
    \item Update $P_n(t)$ using KKT conditions or solving \eqref{eqn_OptPower_neqeq22} if $\frac{T}{\tau_{\rm max}} < Q$;
  \end{enumerate}
  \item Update $P^0_n(t)$ using $P_n(t)$;
\end{itemize}
\item \textbf{End while}
\end{itemize}
\label{Algo}
\end{algorithm}


\section{Cost, Delay and Delivery Content Size Tradeoff Analysis}\label{Tradeoff}

In our system design, our target is to minimize the total network cost subject to the delay constraint and total content delivery constraint. Both constraints actually determine the total cost consumed by the system. In fact, there is inherently a tradeoff between the total network cost, delay and delivery content size. In this section, we present the tradeoff analysis of the system, which may help simplify the design of the system.

We first present the tradeoff between the total network cost and the delay requirement for a given delivery content size.

\begin{prop}\label{proposition1}
In the considered system design problem \eqref{eqn_OptPower_5}, for a given delivery content size $Q$, we have
\begin{itemize}
  \item when $\frac{T}{\tau_{\rm max}} \geq Q$, the total network cost increases as the delay requirement becomes strict.
  \item when $\frac{T}{\tau_{\rm max}} < Q$, if we denote the optimal solution to \eqref{eqn_OptPower_5} as $P^*_n(t)$, which is decomposed into the sum $P^*_{n,1}(t)+P^*_{n,2}(t)$, with $P^*_{n,1}(t)$ being the optimal power term obtained by solving \eqref{eqn_OptPower_6}, there exists a delay requirement region
      \begin{equation}\label{eqn_OptPower_neweq4}
      \begin{aligned}
       {\bm \tau}^{\rm reg}_{\rm max} = \left[\frac{1}{\tau},  \tau_{\rm max}\right],
      \end{aligned}
      \end{equation}
      where $\tau = B\log_2(1+\sum_{n \in \mathcal{N}}  \frac{G P^*_{n,1}(t)}{d^\alpha_n(t)\sigma^2}+ \min_t \sum_{n \in \mathcal{N}}  \frac{G P^*_{n,2}(t)}{d^\alpha_n(t)\sigma^2}  )$ such that the total network cost does not increase with the decrease in delay.
\end{itemize}
\end{prop}
\begin{proof}
Under the condition $\frac{T}{\tau_{\rm max}} \geq Q$, the original system design problem \eqref{eqn_OptPower_5} is equivalent to problem \eqref{eqn_OptPower_6}, where we have only the power constraint \label{eqn_OptPower_6a} and the delay constraint \label{eqn_OptPower_6b}. At the optimal solution, the delay constraint \label{eqn_OptPower_6b} is always active. Hence, if we decrease the value of $\tau_{\rm max}$ to achieve a strict delay requirement, the total network cost must be increased.

If under the condition $\frac{T}{\tau_{\rm max}} < Q$, for a given delay requirement $\tau_{\rm max}$, the optimal solution is $P^*_n(t)$, which is decomposed into the sum $P^*_{n,1}(t)+P^*_{n,2}(t)$, with $P^*_{n,1}(t)$. We first prove that the solution $P^*_n(t)$  with its decomposition $P^*_{n,1}(t)$ and $P^*_{n,2}(t)$ is also the solution to \eqref{eqn_OptPower_5} for a smaller given delay requirement $\tau^\prime_{\rm max} \in {\bm \tau}^{\rm reg}_{\rm max}$. Then, we prove that this solution is the optimal solution.

It is noted that as $B \log_2 \left(1+\sum_{n \in \mathcal{N}} \frac{P^*_{n,1}(t) |h_n(t)|^2}{\sigma^2} \right) = \tau_{\rm max}$, we have $\sum_{n \in \mathcal{N}} \frac{P^*_{n,2}(t) |h_n(t)|^2}{\sigma^2}\geq 0$, as claimed in Lemma \ref{corollary_2}. This further produces
\begin{equation}\label{eqn_OptPower_neweq5}
\begin{aligned}
     &  B\log_2 \left(1+\sum_{n \in \mathcal{N}} \frac{P^*_n(t) |h_n(t)|^2}{\sigma^2} \right)  \\
      \geq & B\log_2 \left(1+ \sum_{n \in \mathcal{N}}  \frac{G P^*_{n,1}(t)}{d^\alpha_n(t)\sigma^2}+ \min_t \sum_{n \in \mathcal{N}}  \frac{G P^*_{n,2}(t)}{d^\alpha_n(t)\sigma^2}  \right)\\
      \geq & \frac{1}{\tau_{\rm max}}.
\end{aligned}
\end{equation}
Thus, if we change the delay requirement $\tau_{\rm max}$ to any smaller value $\tau^\prime_{\rm max}$ in region ${\bm \tau}^{\rm reg}_{\rm max}$, $P^*_n(t)$ can still be a solution to the design problem \eqref{eqn_OptPower_5}, and they have the same total network cost.

In the following, we prove that $P^*_n(t)$ is the optimal solution to the design problem \eqref{eqn_OptPower_5} with a smaller delay requirement $\tau^\prime_{\rm max} \in {\bm \tau}^{\rm reg}_{\rm max}$, $P^*_n(t)$. It is noted that if we decrease the delay requirement $\tau_{\rm max}$ to $\tau^\prime_{\rm max}$, the design problem \eqref{eqn_OptPower_5} with delay requirement $\tau_{\rm max}$ has a smaller feasible region than that of the design problem with delay requirement $\tau^\prime_{\rm max}$. Then, the value of the objective function, i.e., the total network cost, of the former cannot be smaller than that of the latter. This indicates that $P^*_n(t)$ is still the optimal solution to the design problem with delay requirement $\tau^\prime_{\rm max}$, which completes the proof of Proposition \label{proposition2}.
\end{proof}

Next, we discuss the tradeoff between the total network cost and delivery content size for a given delay requirement.

\begin{prop}\label{proposition2}
In the considered system design problem \eqref{eqn_OptPower_5}, for a given delay requirement $\tau_{\rm max}$, we have
\begin{itemize}
  \item when $Q\leq \frac{T}{\tau_{\rm max}}$, increasing the delivery content size will not lead to an increase in the total network cost.
  \item when $Q> \frac{T}{\tau_{\rm max}}$, increasing the delivery content size must also increase the total network cost.
\end{itemize}
\end{prop}
\begin{proof}
Under the condition $Q\leq \frac{T}{\tau_{\rm max}}$, because constraint \eqref{eqn_OptPower_5d} is the same as \eqref{eqn_OptPower_5}, the system cost is determined only by the delay requirement. Therefore, increasing the delivery content size will not lead the total network cost to increase.

Under the condition $Q> \frac{T}{\tau_{\rm max}}$, based on the result presented in Lemma \ref{corollary_2}, with a constant power $P^*_{n,1}(t)$, we have to increase the value of $\sum_{n \in \mathcal{N}} \frac{P^*_{n,2}(t) |h_n(t)|^2}{\sigma^2}$ to meet constraint \eqref{eqn_OptPower_5d} with larger $Q$. This leads to a larger total network cost. It is noted that when increasing the size of delivery content in \eqref{eqn_OptPower_5}, we cannot find a solution to $P^*_{n,2}(t)$ that satisfies constraint \eqref{eqn_OptPower_5d} while not increasing the value of the objective function in \eqref{eqn_OptPower_5}. We prove this conclusion by using a contradiction statement. Assume that with a delivery content size $Q^\prime$, the optimal solution to \eqref{eqn_OptPower_5} is ${P^\prime}^*_n(t)={P^\prime}^*_{n,1}(t)+{P^\prime}^*_{n,2}(t)$. Now, assume that with a larger delivery content size $Q^{\prime \prime}$, the optimal solution to \eqref{eqn_OptPower_5} is ${P^{\prime \prime}}^*_n(t)={P^\prime}^*_{n,1}(t)+{P^{\prime \prime}}^*_{n,2}(t)$. If ${P^\prime}^*_n(t)$ and ${P^{\prime \prime}}^*_n(t)$ have the same objective function value in \eqref{eqn_OptPower_5}, we can always obtain a new solution to \eqref{eqn_OptPower_5} with a delivery content size $Q^\prime$ as $\tilde{{P^\prime}}^*_n(t)={P^\prime}^*_{n,1}(t)+\alpha {P^{\prime \prime}}^*_{n,2}(t)$, where $\alpha$ is a positive value smaller than $1$. This new solution $\tilde{{P^\prime}}^*_n(t)$ produces a smaller objective function value than ${P^\prime}^*_n(t)$. This contradicts the fact that ${P^\prime}^*_n(t)$ is the optimal solution, which completes the proof of Proposition \ref{proposition2}.
\end{proof}

\section{Invariant Power Optimization with QoS Constraints}\label{Invariant}

As another simple power allocation scheme, we consider a constant power optimization design where power does not vary with the channel. Under this situation, the overall optimization problem can be modified as
\begin{subequations}\label{eqn_problem_Invar_1}
\begin{eqnarray}
\min\limits_{P_n}&&  T \sum_{n \in \mathcal{N}} P_n  +  \label{eqn_problem_Invar_1a} \\
&& \int_0^T \sum_{n \in \mathcal{N}} \beta  || P_n||_0 (1-c_{n,f})R_n(t) d t  \nonumber \\
s.t. && 0\leq P_n \leq P_{n,{\rm avg}}~~ \forall n \in \mathcal{N} \label{eqn_problem_Invar_1b}\\
&&  \frac{1}{ C(t)}  \leq \tau_{\rm max} \label{eqn_problem_Invar_1c}\\
&& \int_0^T C(t) dt  \geq Q  \label{eqn_problem_Invar_1d}\\
&&  P_n(t)\geq 0 ~~ \forall n \in \mathcal{N}, \label{eqn_problem_Invar_5f}
\end{eqnarray}
\end{subequations}
where $C(t) = B \log_2 \left(1+\sum_{n \in \mathcal{N}} \frac{G P_n} {d_n(t)^\alpha \sigma^2} \right)$. Similar to the solution we presented in Section \ref{Dynamic},  we determine $P_n$ by solving the following problem:
\begin{subequations}\label{eqn_problem_Invar_12}
\begin{eqnarray}
&&\min\limits_{P_n}  \sum_{n \in \mathcal{N}}k^\prime_n P_n \label{eqn_problem_Invar_12a} \\
s.t. && 0\leq P_n \leq P_{n,{\rm avg}}~~ \forall n \in \mathcal{N} \label{eqn_problem_Invar_12b}\\
&&  \frac{1}{C(t)}  \leq \tau_{\rm max}  \label{eqn_problem_Invar_12c}\\
&& \int_0^T C(t) dt  \geq Q, \label{eqn_problem_Invar_12d}
\end{eqnarray}
\end{subequations}
where $k^\prime_n = T+ \frac{1}{\log(1/\theta+1)} \frac{\beta  (1-c_{n,f}) \int_0^T R_n(t) d t}{\theta+P^0_n} $, with $P^0_n$ being a basis point of the Taylor expansion.

To solve \eqref{eqn_problem_Invar_12}, two specific cases are considered. If $\frac{T}{\tau_{\rm max}} \geq Q$, we have $C(t)=\frac{1}{\tau_{\rm max}}$. Then, constraint \eqref{eqn_problem_Invar_12d} is redundant. $P_n$ can be found by solving
\begin{subequations}\label{eqn_problem_Invar_2}
\begin{eqnarray}
&&\min\limits_{P_n}  \sum_{n \in \mathcal{N}}k^\prime_n P_n \label{eqn_problem_Invar_2a} \\
s.t. && 0\leq P_n \leq P_{n,{\rm avg}}~~ \forall n \in \mathcal{N} \label{eqn_problem_Invar_2b}\\
&&  \sum_{n \in \mathcal{N}} \frac{G P_n} {d_n(t)^\alpha \sigma^2}  \geq 2^{\frac{1}{\tau_{\rm max} B}}-1. \label{eqn_problem_Invar_2c}
\end{eqnarray}
\end{subequations}
It is noted in \eqref{eqn_problem_Invar_2} that constraint \eqref{eqn_problem_Invar_2c} should be satisfied for any arbitrary $t$, and thus causes \eqref{eqn_problem_Invar_2} to contain infinite constraints.
We next solve \eqref{eqn_problem_Invar_2} as problem \eqref{eqn_OptPower_neqeq22} by sampling the time period at certain discrete time points. Denote the discrete time points $\{t_1, t_2,\cdots,t_M\}$ and the adjacent sampling point interval as ${\bigtriangleup t}$; the power $P_n$ can be efficiently obtained through solving the following linear programming problem:
\begin{subequations}\label{eqn_problem_Invar_222}
\begin{eqnarray}
&&\min\limits_{P_{n,1}}  \sum_{n \in \mathcal{N}}k^\prime_n P_n \label{eqn_problem_Invar_222a} \\
s.t. && 0\leq P_n \leq P_{n,{\rm avg}}~~ \forall n \in \mathcal{N} \label{eqn_problem_Invar_222b}\\
&&  \sum_{n \in \mathcal{N}} \frac{G P_n} {d_n(t_m)^\alpha \sigma^2}  \geq 2^{\frac{1}{\tau_{\rm max} B}}-1, \forall i. \label{eqn_problem_Invar_222c}
\end{eqnarray}
\end{subequations}

If $\frac{T}{\tau_{\rm max}} < Q$, similar to its dynamic counterpart, the power can be denoted as $P_n = P_{n,1}+P_{n,2}$ where $P_{n,1}$ is used to activate the constraint \eqref{eqn_problem_Invar_222c} in \eqref{eqn_problem_Invar_222}. Then, $P_{n,2}$ can be obtained
by using the Lagrangian method or time period sampling approach.

\section{Numerical Results}\label{simulation}

In the following, we provide numerical results to show the superiority of using dynamic power allocation in a high-speed moving transmission scenario. In particular, we compare the performance of dynamic power allocation and invariant power optimization with respect to different caching schemes. The parameter settings for our simulation are summarized in TABLE \ref{table1}. Moreover, we consider three caching strategies to illustrate the effect of caching on the performance, that is, PopC, RndC and the noncaching scheme (NonC). Because the PopC caching strategy asks each RRH to cache the most popular contents until its storage is full, based on our parameter settings, RRH $1$ and RRH $2$ both cache contents $\{1,2,3,4,5\}$. NonC assumes that the RRH has no storage resources and that no content is cached at the RRH.

\begin{table}[h]
\scriptsize
\centering
\caption{Parameter settings in simulations}\label{table1}
\begin{tabular}{|l|l|l|}
\hline
Parameter & Notation & Value \\ \hline
Height of RRH antennas & $h$ & $20 $ m  \\ \hline
Interval between two RRHs & $d$ & $1000 $ m \\ \hline
Distance between RRH and road & $d_0$ & $100 $ m  \\ \hline
Coordinates of RRH $1$ & $(l_1, d_0)$ & $(-200 ~\rm{m}, 100 ~\rm{m})$ \\ \hline
Coordinates of RRH $2$ & $(l_2, d_0)$ & $(800 ~\rm{m}, 100 ~\rm{m})$ \\ \hline
Path-loss exponent & $\alpha$ & $0.8$  \\ \hline
Channel gain & $G$ & $2$ \\ \hline
Train speed & $v_0$ & $200 ~\rm{Km/h}$ \\ \hline
Ratio of backhaul power cost and rate & $\beta$ & $2.8$ \\ \hline
Noise power & $\sigma^2$ & $1$ \\ \hline
Content number & $L$ & $15$  \\ \hline
Content size & $Q$ & $1$ \\ \hline
RRH storage size & $F_n$ & $5$ \\ \hline
Shaping parameter of content popularity & $\eta$ &  $1$ \\ \hline
\end{tabular}
\end{table}

In Fig.~\ref{Power_Varying}, the convergence behavior of the proposed algorithm is illustrated. We can observe that the proposed algorithm converges fast in no more than five iterations. Three curves with different delay requirements also demonstrate that a decrease in delay can significantly increase the total network power cost.

In Fig. \ref{Channel_and_power}, the dynamic power is illustrated with the time-varying channel. As the train departs from RRH 1 and moves closer to RRH 2, we see that the channel gain of $h_1(t)$ decreases as time passes, while the channel gain of $h_2(t)$ increases with time. To satisfy the QoS requirements, the power at RRH $1$ gradually increases to compensate for the channel gain loss of $h_1(t)$ while the power at RRH $2$ remains zero. As the train moves closer to RRH $2$, RRH $2$ begins to serve it, and the power at RRH $1$ becomes zero to maintain a lower network power cost. In particular, when the train gets close to RRH $2$, the required power at RRH $2$ gradually decreases as the quality of channel $h_1(t)$ improves.

\begin{figure}
 	\centering
 	\includegraphics[width= 3 in]{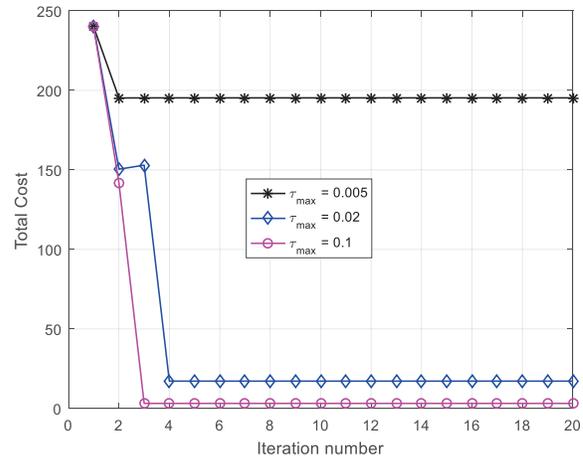}  \vspace{-2mm}
 	\caption{Convergence behavior of the proposed Algorithm 1 at an average SNR $\frac{P_{1,{\rm avg}}}{\sigma^2} = \frac{P_{2,{\rm avg}}}{\sigma^2}= 10$ dB.} \label{Power_Varying}
 \end{figure}

 \begin{figure}
 	\centering
 	\includegraphics[width= 3 in]{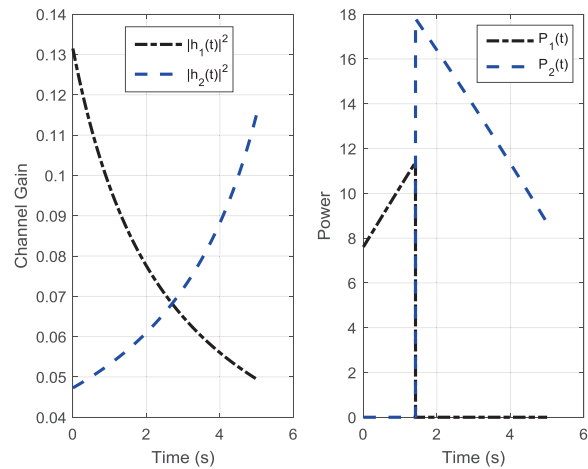}  \vspace{-2mm}
 	\caption{Power variance with a dynamic channel.} \label{Channel_and_power}
 \end{figure}

  \begin{figure}
 	\centering
 	\includegraphics[width= 3 in]{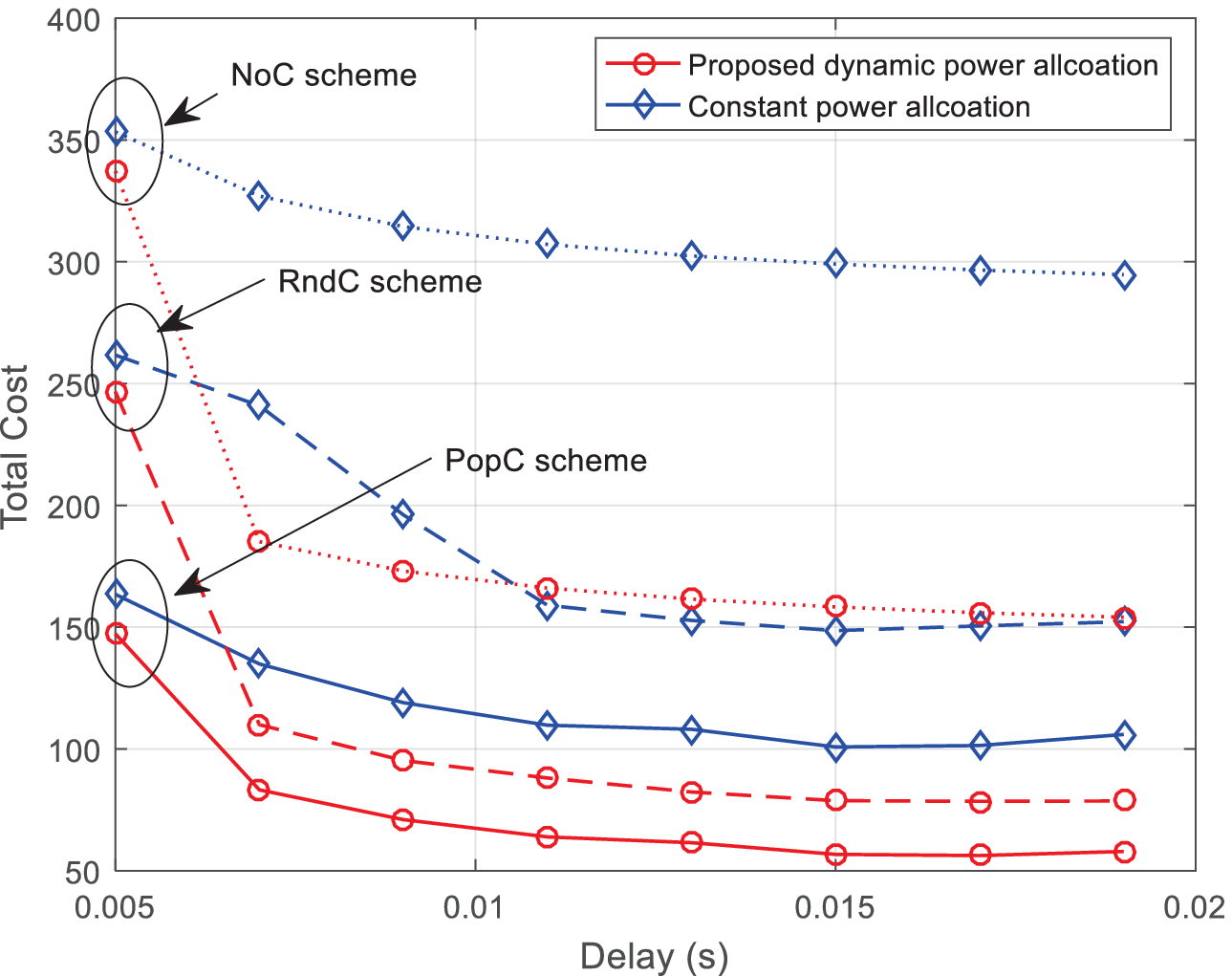}  \vspace{-2mm}
 	\caption{Total power cost for different delay requirements at an average SNR $\frac{P_{1,{\rm avg}}}{\sigma^2} = \frac{P_{2,{\rm avg}}}{\sigma^2}= 10$ dB.} \label{Change_delay}
 \end{figure}

  \begin{figure}
 	\centering
 	\includegraphics[width= 3 in]{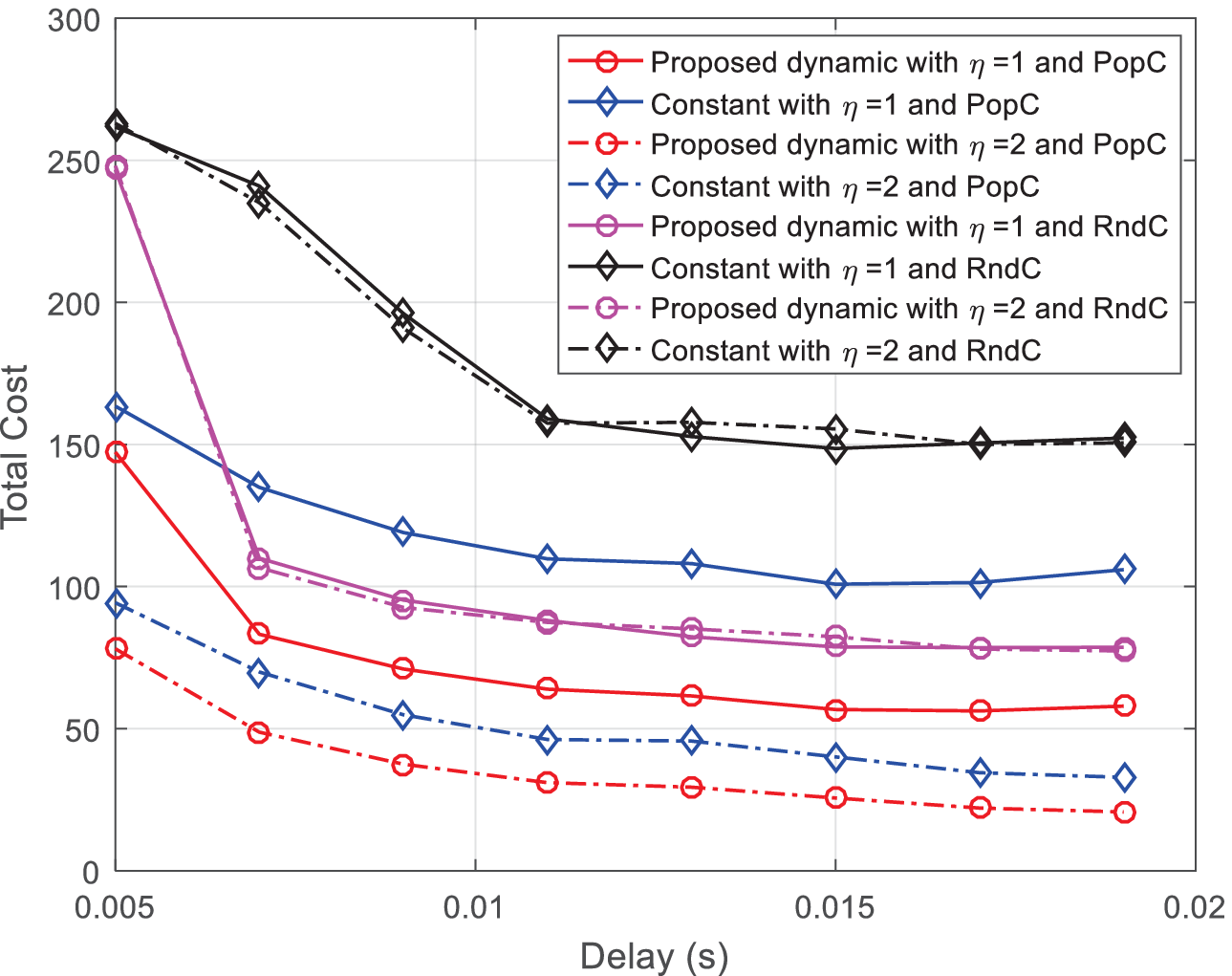}  \vspace{-2mm}
 	\caption{Total power cost for different delay requirements and $\eta$ at an average SNR $\frac{P_{1,{\rm avg}}}{\sigma^2} = \frac{P_{2,{\rm avg}}}{\sigma^2}= 10$ dB.} \label{Change_eta}
 \end{figure}

  \begin{figure}
 	\centering
 	\includegraphics[width= 3 in]{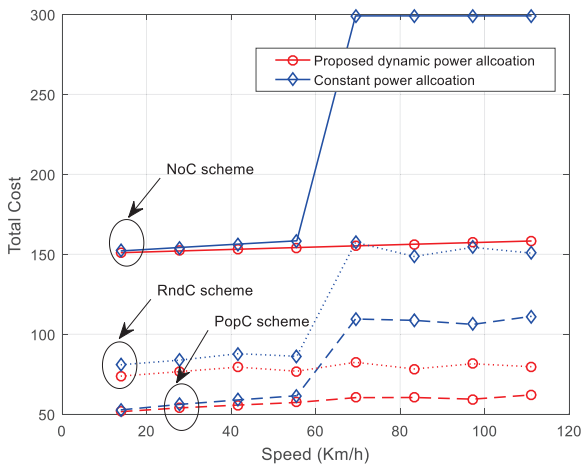}  \vspace{-2mm}
 	\caption{Total power cost for different train speeds at an average SNR $\frac{P_{1,{\rm avg}}}{\sigma^2} = \frac{P_{2,{\rm avg}}}{\sigma^2}= 10$ dB.} \label{Change_speed}
 \end{figure}

In Fig. \ref{Change_delay}, we illustrate the effect of the delay requirements on the total network power cost for different caching strategies. It is observed that a stricter delay requirement enhances the total network power cost and that the proposed dynamic power allocation significantly outperforms the invariant power allocation. Moreover, the PopC caching strategy has the lowest total power cost. The RndC strategy performs worse than the PopC strategy. The NonC scheme has the maximum total power cost. The result implies that caching at the RRH is beneficial for reducing the network power consumption and that caching popular contents is more helpful.

Fig. \ref{Change_eta} further illustrates the effect of the shaping parameter of popularity $\eta$. In general, a larger $\eta$ implies a larger popularity difference among the contents. We observe that a larger $\eta$ produces a lower total network power for the PopC caching strategy. This observation is reasonable as the request contents are very likely to be cached at the RRHs, which can reduce the backhaul power consumption. However, for the RndC strategy, we see that the change in $\eta$ has little effect on the total network power cost. The main reason is that the RndC strategy does not consider the content popularity and caches all contents with equal probabilities.

In Fig. \ref{Change_speed}, we show the effect of train speed on the total network power consumption. It is observed that the total network power cost increases with the speed, especially for dynamic power allocation. In the given time period $(0, T]$, a higher speed implies that a longer distance will be covered by the train. This increases the total power consumption.

\section{Conclusions}\label{conclusion}
In this paper, we studied the dynamic power allocation of the Fog-RAN-assisted high-speed railway system. For a given caching strategy, we optimized the instantaneous power allocation at the RRHs with the aim to minimize the network power consumption subject in total to several QoS constraints. By analyzing the dynamic power optimization problem, we derived the analytical power solution. Our results showed that caching at the RRHs can significantly reduce the total network power consumption. More so, the dynamic power allocation is significantly superior to the invariant one, as it takes the time-varying characteristic of the channel into consideration.

\appendices
\section{Feasibility analysis of problem \eqref{eqn_OptPower_6new3}}

\begin{lemma}\label{corollary_appendix}
For problem \eqref{eqn_OptPower_6new3}, if $\tilde{t}^\prime> t^{\prime  \prime}$, the optimization problem is infeasible.
\end{lemma}
\begin{proof}
With an assumption $\tilde{t}^\prime> t^{\prime  \prime}$, we have $\int_{t^{\prime  \prime}}^T \tilde{a}_3(t) dt >  T P_{2,{\rm avg}}$.
Next, we prove the infeasibility of problem \eqref{eqn_OptPower_6new3} by contradiction analysis. Assume that we have a feasible solution $P^\prime_2(t)$ that satisfies $P^\prime_2(t)< \tilde{a}_3(t)$ for $t\in (t^{\prime  \prime}, T]$ and
\begin{equation}\label{eqn_appendix}
\int_0^{t^{\prime  \prime}} \tilde{a}_2 (t) P^\prime_2(t) dt + \int_{t^{\prime  \prime}}^T \tilde{a}_2 (t) P^\prime_2(t) dt = \int_{t^{\prime  \prime}}^T \tilde{a}_2 (t) \tilde{a}_3(t) dt = B.
\end{equation}
Because $\tilde{a}_2 (t)$ is an increasing function, when \eqref{eqn_appendix} is satisfied, although $\int_{t^{\prime  \prime}}^T P^\prime_2(t)dt <  T P_{2,{\rm avg}}$, we have
\begin{equation}\label{eqn_appendix_1}
\begin{split}
\int_0^{t^{\prime  \prime}} P^\prime_2(t)dt & > \int_{t^{\prime  \prime}}^T \tilde{a}_3(t) dt -\int_{t^{\prime  \prime}}^T P^\prime_2(t)dt \\
& > T P_{2,{\rm avg}} - \int_{t^{\prime  \prime}}^T P^\prime_2(t)dt,
\end{split}
\end{equation}
which implies that $P^\prime_2(t)$ cannot be a feasible solution. We thus complete the proof of Lemma \ref{corollary_appendix}.
\end{proof}


\bibliographystyle{IEEEtran}
\bibliography{IEEEabrv,Power_HSR_FogRAN}

\begin{thebibliography}{10}
\providecommand{\url}[1]{#1}
\csname url@samestyle\endcsname
\providecommand{\newblock}{\relax}
\providecommand{\bibinfo}[2]{#2}
\providecommand{\BIBentrySTDinterwordspacing}{\spaceskip=0pt\relax}
\providecommand{\BIBentryALTinterwordstretchfactor}{4}
\providecommand{\BIBentryALTinterwordspacing}{\spaceskip=\fontdimen2\font plus
\BIBentryALTinterwordstretchfactor\fontdimen3\font minus
  \fontdimen4\font\relax}
\providecommand{\BIBforeignlanguage}[2]{{%
\expandafter\ifx\csname l@#1\endcsname\relax
\typeout{** WARNING: IEEEtran.bst: No hyphenation pattern has been}%
\typeout{** loaded for the language `#1'. Using the pattern for}%
\typeout{** the default language instead.}%
\else
\language=\csname l@#1\endcsname
\fi
#2}}
\providecommand{\BIBdecl}{\relax}
\BIBdecl

\bibitem{Zhao_JSAC_2016}
Z.~{Zhao}, M.~{Peng}, Z.~{Ding}, W.~{Wang}, and H.~V. {Poor}, ``Cluster content
  caching: An energy-efficient approach to improve quality of service in cloud
  radio access networks,'' \emph{IEEE Journal on Selected Areas in
  Communications}, vol.~34, no.~5, pp. 1207--1221, 2016.

\bibitem{Yan_access2018}
D.~{Yan}, R.~{Wang}, E.~{Liu}, and Q.~{Hou}, ``Admm-based robust beamforming
  design for downlink cloud radio access networks,'' \emph{IEEE Access},
  vol.~6, pp. 27\,912--27\,922, 2018.

\bibitem{Pen_access2016}
M.~{Peng} and K.~{Zhang}, ``Recent advances in fog radio access networks:
  Performance analysis and radio resource allocation,'' \emph{IEEE Access},
  vol.~4, pp. 5003--5009, 2016.

\bibitem{Zhao_WCM_2020}
Z.~{Zhao}, C.~{Feng}, H.~H. {Yang}, and X.~{Luo}, ``Federated-learning-enabled
  intelligent fog radio access networks: Fundamental theory, key techniques,
  and future trends,'' \emph{IEEE Wireless Communications}, vol.~27, no.~2, pp.
  22--28, 2020.

\bibitem{Jiang_TWC_2020}
Y.~{Jiang}, C.~{Wan}, M.~{Tao}, F.~C. {Zheng}, P.~{Zhu}, X.~{Gao}, and
  X.~{You}, ``Analysis and optimization of fog radio access networks with
  hybrid caching: Delay and energy efficiency,'' \emph{IEEE Transactions on
  Wireless Communications}, pp. 1--1, 2020.

\bibitem{Wang_TVT2019}
R.~{Wang}, R.~{Li}, P.~{Wang}, and E.~{Liu}, ``Analysis and optimization of
  caching in fog radio access networks,'' \emph{IEEE Transactions on Vehicular
  Technology}, vol.~68, no.~8, pp. 8279--8283, 2019.

\bibitem{Wang_CL_2019}
R.~{Wang}, R.~{Li}, E.~{Liu}, and P.~{Wang}, ``Performance analysis and
  optimization of caching placement in heterogeneous wireless networks,''
  \emph{IEEE Communications Letters}, vol.~23, no.~10, pp. 1883--1887, 2019.

\bibitem{Bai_2019_access}
W.~{Bai}, T.~{Yao}, H.~{Zhang}, and V.~C.~M. {Leung}, ``Research on channel
  power allocation of fog wireless access network based on noma,'' \emph{IEEE
  Access}, vol.~7, pp. 32\,867--32\,873, 2019.

\bibitem{Rahman_TVT_2020}
G.~M.~S. {Rahman}, M.~{Peng}, S.~{Yan}, and T.~{Dang}, ``Learning based joint
  cache and power allocation in fog radio access networks,'' \emph{IEEE
  Transactions on Vehicular Technology}, vol.~69, no.~4, pp. 4401--4411, 2020.

\bibitem{Xiang_TVT_2020}
H.~{Xiang}, M.~{Peng}, Y.~{Sun}, and S.~{Yan}, ``Mode selection and resource
  allocation in sliced fog radio access networks: A reinforcement learning
  approach,'' \emph{IEEE Transactions on Vehicular Technology}, vol.~69, no.~4,
  pp. 4271--4284, 2020.

\bibitem{Zhang_AT-RASC2018}
H.~{Zhang}, L.~{Zhu}, K.~{Long}, and X.~{Li}, ``Energy efficient resource
  allocation in millimeter-wave-based fog radio access networks,'' in
  \emph{2018 2nd URSI Atlantic Radio Science Meeting (AT-RASC)}, 2018, pp.
  1--4.

\bibitem{He_TCOM_2019}
S.~{He}, C.~{Qi}, Y.~{Huang}, Q.~{Hou}, and A.~{Nallanathan}, ``Two-level
  transmission scheme for cache-enabled fog radio access networks,'' \emph{IEEE
  Transactions on Communications}, vol.~67, no.~1, pp. 445--456, 2019.

\bibitem{Tao_TWC_2016}
M.~{Tao}, E.~{Chen}, H.~{Zhou}, and W.~{Yu}, ``Content-centric sparse multicast
  beamforming for cache-enabled cloud {RAN},'' \emph{IEEE Transactions on
  Wireless Communications}, vol.~15, no.~9, pp. 6118--6131, 2016.

\bibitem{ErkaiTWC_2018}
E.~{Chen}, M.~{Tao}, and Y.~{Liu}, ``Joint base station clustering and
  beamforming for non-orthogonal multicast and unicast transmission with
  backhaul constraints,'' \emph{IEEE Transactions on Wireless Communications},
  vol.~17, no.~9, pp. 6265--6279, 2018.

\bibitem{Ma_Acess_2020}
Y.~{Ma}, H.~{Wang}, J.~{Xiong}, J.~{Diao}, and D.~{Ma}, ``Joint allocation on
  communication and computing resources for fog radio access networks,''
  \emph{IEEE Access}, vol.~8, pp. 108\,310--108\,323, 2020.

\bibitem{Khumalo_2020}
N.~{Khumalo}, O.~{Oyerinde}, and L.~{Mfupe}, ``Reinforcement learning-based
  computation resource allocation scheme for 5g fog-radio access network,'' in
  \emph{2020 Fifth International Conference on Fog and Mobile Edge Computing
  (FMEC)}, 2020, pp. 353--355.

\bibitem{Li_2020_TWC}
K.~{Li}, M.~{Tao}, and Z.~{Chen}, ``Exploiting computation replication for
  mobile edge computing: A fundamental computation-communication tradeoff
  study,'' \emph{IEEE Transactions on Wireless Communications}, vol.~19, no.~7,
  pp. 4563--4578, 2020.

\bibitem{Dang_jSAC_2019}
T.~{Dang} and M.~{Peng}, ``Joint radio communication, caching, and computing
  design for mobile virtual reality delivery in fog radio access networks,''
  \emph{IEEE Journal on Selected Areas in Communications}, vol.~37, no.~7, pp.
  1594--1607, 2019.

\bibitem{Ai_ComM2015}
B.~{Ai}, K.~{Guan}, M.~{Rupp}, T.~{Kurner}, X.~{Cheng}, X.~{Yin}, Q.~{Wang},
  G.~{Ma}, Y.~{Li}, L.~{Xiong}, and J.~{Ding}, ``Future railway
  services-oriented mobile communications network,'' \emph{IEEE Communications
  Magazine}, vol.~53, no.~10, pp. 78--85, 2015.

\bibitem{Fan_Access2016}
J.~{Wu} and P.~{Fan}, ``A survey on high mobility wireless communications:
  Challenges, opportunities and solutions,'' \emph{IEEE Access}, vol.~4, pp.
  450--476, 2016.

\bibitem{Muneer_TWC2015}
P.~{Muneer} and S.~M. {Sameer}, ``Joint {ML} estimation of {CFO} and channel,
  and a low complexity turbo equalization technique for high mobility {OFDMA}
  uplinks,'' \emph{IEEE Transactions on Wireless Communications}, vol.~14,
  no.~7, pp. 3642--3654, 2015.

\bibitem{Wang_JSAC2012}
J.~{Wang}, H.~{Zhu}, and N.~J. {Gomes}, ``Distributed antenna systems for
  mobile communications in high speed trains,'' \emph{IEEE Journal on Selected
  Areas in Communications}, vol.~30, no.~4, pp. 675--683, 2012.

\bibitem{Li_Access2017}
T.~{Li}, K.~{Xiong}, P.~{Fan}, and K.~B. {Letaief}, ``Service-oriented power
  allocation for high-speed railway wireless communications,'' \emph{IEEE
  Access}, vol.~5, pp. 8343--8356, 2017.

\bibitem{Zhang_TVT2015}
C.~{Zhang}, P.~{Fan}, K.~{Xiong}, and P.~{Fan}, ``Optimal power allocation with
  delay constraint for signal transmission from a moving train to base stations
  in high-speed railway scenarios,'' \emph{IEEE Transactions on Vehicular
  Technology}, vol.~64, no.~12, pp. 5775--5788, 2015.

\bibitem{Liu_Access2018}
X.~{Liu} and D.~{Qiao}, ``Location-fair beamforming for high speed railway
  communication systems,'' \emph{IEEE Access}, vol.~6, pp. 28\,632--28\,642,
  2018.

\bibitem{Wang_TVT_2019}
R.~{Wang}, R.~{Li}, P.~{Wang}, and E.~{Liu}, ``Analysis and optimization of
  caching in fog radio access networks,'' \emph{IEEE Transactions on Vehicular
  Technology}, vol.~68, no.~8, pp. 8279--8283, 2019.

\bibitem{Simonetto_Proceedings_2020}
A.~{Simonetto}, E.~{Dall'Anese}, S.~{Paternain}, G.~{Leus}, and G.~B.
  {Giannakis}, ``Time-varying convex optimization: Time-structured algorithms
  and applications,'' \emph{Proceedings of the IEEE}, pp. 1--17, 2020.

\bibitem{cvx}
M.~Grant and S.~Boyd, ``{CVX}: Matlab software for disciplined convex
  programming, version 2.1,'' \url{http://cvxr.com/cvx}, Mar. 2014.

\end{thebibliography}
\end{document}